\documentclass[12pt]{article}

\usepackage[T1]{fontenc}
\usepackage{lmodern}
\usepackage[margin=1in]{geometry}
\usepackage{amsmath,amssymb,amsthm,mathtools}
\usepackage{booktabs}
\usepackage{graphicx}
\usepackage{microtype}
\usepackage{enumitem}
\usepackage{xcolor}
\usepackage{hyperref}
\hypersetup{
  colorlinks=true,
  linkcolor=blue!45!black,
  citecolor=blue!45!black,
  urlcolor=blue!45!black
}
\usepackage{tikz}
\usetikzlibrary{positioning,arrows.meta}

\newcommand{\R}{\mathbb R}

\newcommand{\PGF}{\chi}
\newcommand{\Hsh}{H}
\newcommand{\Sfac}{S_{\mathrm{fac}}}

\newcommand{\supp}{\operatorname{supp}}
\newcommand{\Res}{\operatorname{Res}}

\newtheorem{theorem}{Theorem}[section]
\newtheorem{proposition}[theorem]{Proposition}
\newtheorem{corollary}[theorem]{Corollary}

\theoremstyle{definition}
\newtheorem{definition}[theorem]{Definition}
\newtheorem{example}[theorem]{Example}

\theoremstyle{remark}
\newtheorem{remark}[theorem]{Remark}

\title{On Factorizing Aggregate Counting Distributions into Independent Latent Processes}

\author{
Israel Klich\\[0.5em]
\small Department of Physics, University of Virginia, Charlottesville, Virginia 22904, USA\\
\small Max Planck Institute for the Physics of Complex Systems, Dresden, Germany\\
\small \texttt{ik3j@virginia.edu}
}

\date{}

\begin{document}

\maketitle

\begin{abstract}
Given only the probability distribution of an aggregate counting variable, what independent latent counting processes are compatible with the observation? Equivalently, when does a probability-generating function admit a factorization into normalized polynomials with nonnegative coefficients? We develop a mathematical theory of such positive factorizations.

We introduce the positive factorization poset, whose elements are all positive factorizations ordered by refinement, and define the factorization entropy, measuring the maximal latent Shannon entropy compatible with the observed distribution. We prove a sharp entropy inequality, characterize the equality case by injectivity of the latent addition map, show that entropy optimization may be restricted to maximal atomizations, and exhibit examples where distinct maximal atomizations have different entropy.

We further establish support-based obstructions to positive factorization, characterize the real-rooted and Hurwitz-stable sectors, prove a local stability theorem for coprime factorizations, determine exactly the factorable regions in degrees two and three, obtain an exact quartic Hurwitz volume, and investigate the geometry of the factorable region inside the probability simplex through exact calculations and Monte Carlo experiments. These results identify the positive factorization poset as a natural algebraic object associated with probability-generating functions and provide a framework for studying latent independent structure in aggregate counting statistics.
\end{abstract}


\medskip


\section{Introduction}
Probability-generating functions provide a compact representation of
discrete probability distributions and play a central role in probability,
statistics, combinatorics, and statistical physics.  Whenever only an
aggregate counting observable (i.e. non-negative integer valued measurement) is experimentally accessible, the
probability-generating function becomes the primary object describing the
measurement:
\begin{equation}
    \PGF_p(z)=\sum_{n=0}^{N}p_nz^n,
    \qquad p_n\ge0,
    \qquad \PGF_p(1)=1.
    \label{eq:pgf}
\end{equation}
If $X=X_1+\cdots+X_m$ is a sum of independent nonnegative integer-valued random variables, then
\begin{equation}
    \PGF_X(z)=\prod_{j=1}^m\PGF_{X_j}(z).
    \label{eq:conv-product}
\end{equation}
Conversely, every factorization of \eqref{eq:pgf} into normalized polynomials with nonnegative coefficients defines such an independent latent decomposition.  The elementary question addressed here is therefore:
\begin{quote}
Which independent latent counting processes are compatible with a given observed distribution?
\end{quote}

Positive coefficient factorizations have a substantial history in ligand-binding theory.  Wyman introduced binding polynomials as generating functions for the numbers of occupied sites, and Briggs formalized the notions of a positive polynomial, positive factorization, and $p$-irreducibility \cite{Wyman1965,Briggs1985}.  In that setting a positive factorization corresponds to a decomposition of binding sites into independent modules.  Briggs obtained complete root-geometric classifications in degrees three and four and related irreducibility to cooperativity.  The same algebraic problem arises whenever only an aggregate count is measured and the elementary counting channels are unresolved.

\paragraph{Physical motivation}
The inverse problem described above appears naturally in several branches of
statistical physics in which only an aggregate counting observable is
experimentally accessible while the underlying microscopic counting events
remain unresolved.

One prominent example is the full counting statistics of electrons in
mesoscopic conductors.  There the experimentally measured object is the
probability distribution of the total charge transmitted during a fixed
measurement interval.  For noninteracting fermions the corresponding
probability-generating function factorizes into elementary Bernoulli factors $\chi(z)=\prod_{j=1}^N(1-\nu_j+\nu_jz)$, $0<\nu_j<1$,
reflecting the decomposition of transport into statistically independent
single-particle transmission events \cite{Levitov1996,Peschel2003,Klich2006}.  This structure is the origin for the relation
between charge counting statistics and entanglement entropy in free-fermion
systems 
\cite{KlichLevitov2009,KlichRefaelSilva2006,SongFlindtRachelKlichLeHur2011,AbanovIvanov2009}.
Beyond the free-fermion setting such a factorization generally ceases to
exist, motivating the question of how much independent latent structure can
still be inferred from the measured counting distribution alone.

A closely related inverse problem appears in modern photon-counting
experiments.  Transition-edge-sensor and multiplexed avalanche-photodiode
detectors accurately resolve the total number of detected photons while
generally providing no information about the temporal, spatial, spectral,
or Schmidt modes from which those photons originated.  Consequently, the
experiment again provides a single probability-generating function of the
total count.  Positive factorization of this generating function admits a
natural interpretation as a decomposition into statistically independent
latent optical counting processes, whereas positive irreducibility excludes
such an interpretation within the assumed counting model.  This viewpoint is
particularly relevant to multimode squeezed light, frequency-comb sources,
and photon-number-resolving detectors
\cite{Fitch2003,Sridhar2014,Nehra2017,Eaton2022}.


Rather than concentrating on any one of these applications individually,
our goal is to develop the common mathematical theory underlying them.
The principal new object introduced here is the positive factorization
poset associated with a probability-generating function. We call this the \emph{positive factorization poset}.  On this poset we define an information-theoretic functional,
\begin{equation}
    \Sfac(p)=\max_{\PGF_p=\prod_j\PGF_j}\sum_j \Hsh(\PGF_j),
    \label{eq:Sfac-intro}
\end{equation}
where $\Hsh(\PGF_j)$ is the Shannon entropy of the coefficient vector of the factor.  The excess
\begin{equation}
    \Sfac(p)-\Hsh(p)
\end{equation}
has an operational interpretation as the largest uncertainty about independent latent outcomes that remains after their sum has been observed.

Our main contributions are as follows.
\begin{enumerate}[leftmargin=2em]
\item We introduce the positive factorization poset associated with a
probability-generating function.  We show that maximal positive
atomizations need not be unique and, in general, need not admit a
common refinement.

\item We introduce factorization entropy as an information-theoretic
functional on the positive factorization poset.  We prove an entropy
inequality, characterize its equality condition by injectivity of the
latent addition map, and show that the entropy optimization may be
restricted to maximal atomizations.  We further show that different
maximal atomizations may carry different entropy.

\item We establish several structural results for positive probability
polynomials, including support obstructions, complete
characterizations of the real-rooted and Hurwitz-stable sectors, and a
local stability theorem for coprime positive factorizations.

\item We determine exactly the factorable regions in degrees two and
three, investigate arithmetic families arising from mixed-radix
decompositions, and distinguish equality cases of the entropy
inequality from entropy-maximizing factorizations.

\item We investigate the geometry of the factorable region inside the
probability simplex, provide Monte Carlo estimates through degree
fourteen, and formulate several open structural, asymptotic and
algorithmic questions.
\end{enumerate}

The remainder of the paper is organized as follows.
Section~2 introduces positive probability polynomials and positive
factorizations.
Section~3 develops the factorization poset and defines factorization
entropy.
Sections~4--6 establish the principal structural theorems and investigate
the geometry of the factorable region.
Section~7 discusses random probability polynomials, and numerical estimates.
Section~8 is gives a brief list of open problems

\section{Positive probability polynomials}

\begin{definition}[Probability polynomial]
A polynomial
\begin{equation}
    P(z)=\sum_{n=0}^{d}q_nz^n
\end{equation}
is a \emph{probability polynomial} if $q_n\ge0$ and $P(1)=1$.  It is \emph{strictly positive} if all coefficients are strictly positive.
\end{definition}

Briggs calls a real polynomial positive when its leading and constant coefficients are positive and all intermediate coefficients are nonnegative \cite{Briggs1985}.  Normalization at $z=1$ converts every such polynomial into a probability polynomial without changing its factorization structure.

\begin{definition}[Positive factorization and atom]
A \emph{positive factorization} of a probability polynomial $P$ is an identity
\begin{equation}
    P=P_1\cdots P_m
    \label{eq:positive-factorization}
\end{equation}
with each $P_j$ a nonconstant probability polynomial.  A probability polynomial that admits no such factorization with $m\ge2$ is called a \emph{positive atom}.  Equivalently, it is $p$-irreducible in Briggs's terminology.
\end{definition}

Monomial factors carry no randomness.  Unless explicitly stated otherwise, we regard deterministic shifts as trivial and assume $P(0)>0$.

\begin{definition}[Refinement order]
Let
\[
F=(P_1,\ldots,P_m),\qquad
G=(Q_1,\ldots,Q_r)
\]
be positive factorizations of the same probability polynomial $P$.
We write
\[
F\preceq G
\]
if $G$ is obtained from $F$ by positively factorizing one or more of the
factors of $F$.
The resulting partially ordered set will be denoted
\[
\mathcal F_+(P).
\]
Its minimal element is the trivial one-factor decomposition $(P)$.
Maximal elements are factorizations into positive atoms.
\end{definition}

\begin{example}[Nonunique maximal atomizations]
\label{ex:C5}
The uniform law on $\{0,\ldots,5\}$ satisfies
\[
1+z+z^2+z^3+z^4+z^5
=
(1+z)(1+z^2+z^4)
=
(1+z+z^2)(1+z^3).
\]
After normalization these become two distinct maximal positive
factorizations of the same probability polynomial.
\end{example}

\begin{proposition}[Incomparable maximal atomizations]
\label{prop:nonlattice}
The positive factorization poset need not possess a unique maximal
element.
In general two maximal atomizations need not admit a common positive
refinement.
Consequently $\mathcal F_+(P)$ is not, in general, a lattice.
\end{proposition}

\begin{proof}
Consider the polynomial of Example~\ref{ex:C5}.

The factors $1+z$ and $1+z^3$ are positive atoms by
Corollary~\ref{cor:a+bz^N}, since their
supports are additively indecomposable.

The polynomial $1+z+z^2$ has two nonreal roots and therefore cannot
factor into positive linear factors.

It remains to consider
\[
1+z^2+z^4.
\]
Suppose
\[
1+z^2+z^4=(a+bz^2)(c+dz^2),
\qquad
a,b,c,d>0.
\]
Comparing coefficients gives
\[
ac=1,\qquad
bd=1,\qquad
ad+bc=1.
\]
However,
\[
ad+bc
\ge
2\sqrt{abcd}
=
2,
\]
contradicting the last equality.
Hence $1+z^2+z^4$ is also a positive atom.

Thus both displayed factorizations are maximal.

If a common positive refinement existed, it would refine both maximal
factorizations, forcing one of the atoms above to admit a further
positive factorization, which is impossible.
Hence no common refinement exists.
\end{proof}
Unlike unique factorization domains, maximal positive atomizations need not admit a common refinement. Thus the positive factorization poset is fundamentally different from the divisor lattice associated with ordinary polynomial or integer factorization. This nonuniqueness is precisely what makes the entropy optimization problem nontrivial.
\begin{remark}
Every positive factorization is obtained by grouping the irreducible real
factors of the polynomial.
Since there are only finitely many conjugation-invariant root
partitions, and normalization is fixed by $P_j(1)=1$ of every factor, the positive factorization poset is finite.
Consequently the maximum defining $S_{\rm fac}$
always exists.
\end{remark}

\section{Factorization entropy}

For a probability polynomial
\[
Q(z)=\sum_n q_nz^n,
\]
let
\[
H(Q)
=
-\sum_n q_n\log q_n ,
\]
with the convention $0\log0=0$.

\begin{definition}[Factorization entropy]
Let $P$ be a probability polynomial.
Its factorization entropy is
\[
S_{\rm fac}(P)
=
\max_{(P_1,\ldots,P_m)\in\mathcal F_+(P)}
\sum_{j=1}^m H(P_j).
\]
Since $\mathcal F_+(P)$ is finite, the maximum is
always attained.
\end{definition}

The following theorem contains the basic entropy inequality and all of its
main consequences.

\begin{theorem}[Entropy dominance and atomization principle]
\label{thm:entropy}
Let
\[
P=P_1P_2\cdots P_m
\]
be a positive factorization, and let
$X_1,\ldots,X_m$
be independent random variables with probability-generating functions
$P_1,\ldots,P_m$.
If
\[
X=\sum_{j=1}^mX_j,
\]
then
\[
\sum_{j=1}^mH(X_j)-H(X)
=
H(X_1,\ldots,X_m\,|\,X)
\ge0.
\]

Equality holds if and only if the addition map
\[
\operatorname{supp}(X_1)\times\cdots\times
\operatorname{supp}(X_m)
\longrightarrow
\mathbb N,
\qquad
(x_1,\ldots,x_m)
\mapsto
\sum_jx_j
\]
is injective.

Consequently:

\begin{enumerate}
\item
Factorization entropy is monotone under refinement:
if
\[
F\preceq G,
\]
then
\[
S(F)\le S(G).
\]
 \item Consequently,
\[
S_{\rm fac}(P)
=
\max_{\substack{
F\in\mathcal F_+(P)\\
F\text{ maximal}
}}
S(F).
\]
In particular, at least one maximal positive atomization attains
\(S_{\rm fac}(P)\).
\end{enumerate}
\end{theorem}

\begin{proof}
Independence gives
\[
H(X_1,\ldots,X_m)
=
\sum_{j=1}^mH(X_j).
\]
Since $X$ is a deterministic function of the tuple,
\[
H(X_1,\ldots,X_m)
=
H(X)
+
H(X_1,\ldots,X_m\,|\,X),
\]
which proves the identity and the inequality.

Equality holds precisely when the tuple
$(X_1,\ldots,X_m)$
is determined by its sum on the support of the joint law, i.e. when the
addition map is injective.

If $G$ is obtained from $F$ by refining a single factor
$Q=RS$,
the inequality just proved gives
\[
H(Q)\le H(R)+H(S),
\]
while all other factor entropies remain unchanged.
Iterating proves monotonicity under arbitrary refinements.

Every factorization admits a maximal refinement, since
$\mathcal F_+(P)$ is finite.  Refinement monotonicity shows that the
entropy of such a maximal refinement is at least that of the original
factorization.  Hence the maximum over all factorizations equals the
maximum over maximal atomizations.
\end{proof}

\begin{example}[Strict inequality]
For two fair Bernoulli variables,
\[
\left(\frac12+\frac12z\right)^2
=
\frac14+\frac12z+\frac14z^2.
\]
The factor entropy equals $2\log2$, whereas the entropy of the observed
sum is
$\frac32\log2$.
The missing
$\frac12\log2$
is precisely
\[
H(X_1,X_2\,|\,X).
\]
\end{example}

\begin{example}[Equality]
Let
\[
P(z)=((1-p)+pz^2)((1-q)+qz).
\]
The supports
$\{0,2\}$ and $\{0,1\}$
sum injectively onto
$\{0,1,2,3\}$.
Hence
\[
H(X+Y)=H(X)+H(Y),
\]
hence this factorization is an equality case of Theorem~\ref{thm:entropy}.
\end{example}

A natural question arises: will two nonequivalent maximal atomizations of a probability polynomial carry the same entropy? The next proposition answers this in the negative.

\begin{proposition}[Maximal atomizations can have different entropies]
\label{prop:different-atomization-entropies}
For \(0<\epsilon<1\), define
\[
R_\epsilon(z)=1-(1-\epsilon)z+z^2
\]
and
\[
P_\epsilon(z)
=(1+z)(1+z+z^2)R_\epsilon(z).
\]
Then \(P_\epsilon\) admits the two maximal positive atomizations
\begin{align}
\mathcal F_\epsilon:\qquad
P_\epsilon(z)
&=(1+z)
 \bigl[1+\epsilon z+(1+\epsilon)z^2
       +\epsilon z^3+z^4\bigr],
\\
\mathcal G_\epsilon:\qquad
P_\epsilon(z)
&=(1+z+z^2)
 \bigl[1+\epsilon z+\epsilon z^2+z^3\bigr].
\end{align}
After normalizing every factor at \(z=1\), their factor entropies satisfy
\[
S(\mathcal G_\epsilon)-S(\mathcal F_\epsilon)
=
\frac13\log(1+\epsilon)
-\frac{\epsilon}{3(1+\epsilon)}\log\epsilon
>0.
\]
Hence maximal positive atomizations need not have equal entropy, and the
maximization in the definition of \(S_{\rm fac}\) is genuinely nontrivial.
\end{proposition}

\begin{proof}
The two identities follow from
\[
(1+z)R_\epsilon(z)
=1+\epsilon z+\epsilon z^2+z^3
\]
and
\[
(1+z+z^2)R_\epsilon(z)
=1+\epsilon z+(1+\epsilon)z^2+\epsilon z^3+z^4.
\]
All displayed coefficients are positive.

The real polynomials \(1+z\), \(1+z+z^2\), and
\(R_\epsilon\) are pairwise coprime, while \(R_\epsilon\) is not a
positive polynomial.  It follows from unique factorization over
\(\mathbb R[z]\) that neither
\((1+z)R_\epsilon\) nor
\((1+z+z^2)R_\epsilon\) admits a nontrivial positive factorization.
Thus both displayed factorizations are maximal.

After normalization, the four factor coefficient vectors are
\[
\left(\frac12,\frac12\right),
\qquad
\frac{(1,\epsilon,1+\epsilon,\epsilon,1)}
     {3(1+\epsilon)},
\]
and
\[
\left(\frac13,\frac13,\frac13\right),
\qquad
\frac{(1,\epsilon,\epsilon,1)}
     {2(1+\epsilon)}.
\]
Substitution into the Shannon entropy gives
\[
S(\mathcal G_\epsilon)-S(\mathcal F_\epsilon)
=
\frac13\log(1+\epsilon)
-\frac{\epsilon}{3(1+\epsilon)}\log\epsilon.
\]
This quantity is strictly positive for \(0<\epsilon<1\).
\end{proof}

The difference
\[
\Delta_{\rm fac}(P)
=
S_{\rm fac}(P)-H(P)
\]
will be called the
\emph{latent factorization entropy}.
It measures the maximum hidden uncertainty about independent latent
outcomes that remains after only their sum has been observed.

\section{Structural results}

\subsection{Support obstructions}

\begin{proposition}[Minkowski support identity]
\label{prop:support}
If $P=QR$ is a positive factorization, then
\begin{equation}
    \supp P=\supp Q+\supp R
    :=\{i+j:i\in\supp Q,\ j\in\supp R\}.
    \label{eq:Minkowski}
\end{equation}
Therefore, if $\supp P$ is not a nontrivial Minkowski sum, then $P$ is a positive atom.
\end{proposition}

\begin{proof}
The coefficient of $z^n$ in $QR$ is $\sum_{i+j=n}q_ir_j$.  Since all summands are nonnegative, this coefficient is positive exactly when at least one pair $(i,j)$ in the two supports satisfies $i+j=n$.
\end{proof}

\begin{corollary}
For every $N\ge1$, the binomial $a+bz^N$ with $a,b>0$ is a positive atom.\label{cor:a+bz^N}
\end{corollary}

\begin{proof}
The set $\{0,N\}$ cannot be represented as a nontrivial Minkowski sum of two subsets each containing zero and a positive integer.
\end{proof}

This recovers, in a purely additive form, the maximal-linkage example $1+z^N$ emphasized in binding-polynomial theory \cite{Briggs1985}.

\subsection{Real-rooted and Hurwitz-stable sectors}

\begin{theorem}[Real-rooted probability polynomials]
\label{thm:real-rooted}
Let $P$ be a degree-$N$ probability polynomial with positive endpoints.  The following are equivalent:
\begin{enumerate}[label=(\roman*)]
\item every zero of $P$ is real;
\item every zero of $P$ is negative;
\item $P$ is the probability-generating polynomial of a Poisson-binomial law,
\begin{equation}
    P(z)=\prod_{j=1}^N(1-\nu_j+\nu_jz),
    \qquad 0<\nu_j<1.
    \label{eq:PB}
\end{equation}
\end{enumerate}
In this case
\begin{equation}
    \Sfac(P)=\sum_{j=1}^Nh(\nu_j),
    \qquad h(x)=-x\log x-(1-x)\log(1-x).
    \label{eq:SfacPB}
\end{equation}
\end{theorem}

\begin{proof}
A polynomial with nonnegative coefficients has no positive real zero.  Thus real-rootedness forces every zero to be negative.  Writing a zero as $r_j<0$ and normalizing the corresponding linear factor at $z=1$ gives
\begin{equation}
    \frac{z-r_j}{1-r_j}=1-\nu_j+\nu_jz,
    \qquad \nu_j=(1-r_j)^{-1}\in(0,1).
\end{equation}
Conversely, \eqref{eq:PB} is real-rooted.  The linear factors are positive atoms, and any other factorization will be further reducible all the way to the Poisson-binomials.  Theorem~\ref{thm:entropy} shows that the fully split factorization maximizes the entropy.
\end{proof}

This characterization is standard in the theory of Poisson-binomial laws and total positivity \cite{AissenEtAl1952,TangTang2023}; it is also naturally situated within the broader theory of stable polynomials \cite{BorceaBranden2009}.

\begin{corollary}[Hurwitz-stable sector]
\label{cor:Hurwitz}
If $P$ is a positive Hurwitz-stable polynomial, then every positive atom has degree one or two and
\begin{equation}
    \Sfac(P)=S_{\mathrm{fac}}^{(\le2)}(P),
\end{equation}
where the right-hand side restricts the factor degrees to at most two.
\end{corollary}

\begin{proof}
A real Hurwitz-stable polynomial factors over $\R$ into linear factors with negative roots and quadratic factors associated with nonreal conjugate pairs in the left half-plane.  Every such linear factor is positive.  A quadratic with positive coefficients is either a product of two positive linear factors or a positive atom.  Briggs showed that this factorization into positive linear and $p$-irreducible quadratic factors is unique \cite{Briggs1985}.
\end{proof}

\subsection{Local stability}

For positive integers $d_1,\ldots,d_m$ with $\sum_i d_i=N$, consider the multiplication map
\begin{equation}
  \Phi:\Delta_{d_1}^{\circ}\times\cdots\times\Delta_{d_m}^{\circ}
  \longrightarrow\Delta_N^{\circ},
  \qquad (P_1,\ldots,P_m)\mapsto\prod_iP_i,
  \label{eq:multiplication-map}
\end{equation}
where $\Delta_d^{\circ}$ denotes probability polynomials with strictly positive coefficients and degree $d$. Recall that two real polynomials are called \emph{coprime} if they have
no common nonconstant factor. Equivalently, they share no real root and
no common complex conjugate pair of roots.

The multiplication map sends a tuple of factors to their product. We next ask whether small perturbations of the product determine unique small perturbations of the factors.
\begin{theorem}[Local stability of transverse factorizations]
\label{thm:local-stability}
Suppose $P=P_1\cdots P_m$ with $P_i\in\Delta_{d_i}^{\circ}$ and the factors are pairwise coprime.  Then the differential of $\Phi$ at $(P_1,\ldots,P_m)$ is invertible.  Consequently, $P$ has a neighborhood in $\Delta_N^{\circ}$ in which every polynomial admits a unique nearby positive factorization of the same ordered degree profile.
\end{theorem}

\begin{proof} 
A tangent vector $\dot P_i$ has degree at most $d_i$ and satisfies $\dot P_i(1)=0$.  If it lies in the kernel of the differential, then
\begin{equation}
  \sum_{i=1}^m \dot P_i\prod_{j\ne i}P_j=0.
  \label{eq:kernel}
\end{equation}
Pick any of the factors $P_i$, we can write: 
$P_i M=\dot P_i\prod_{j\ne i}P_j$, where $M$ is a polynomial with degree $N-d_i$. Since $P_i$ is coprime with all the $P_j$, we conclude that
$P_i$ divides $\dot P_i$.  Since $\deg\dot P_i\le\deg P_i$, we have $\dot P_i=c_iP_i$.  Evaluating at $z=1$ gives $c_i=0$. Therefore, there is no non-trivial solution to  \ref{eq:kernel}.
Hence the differential is injective.  The domain and target have the same dimension $\sum_i d_i=N$, so it is invertible.  The inverse function theorem gives a local diffeomorphism, and strict coefficient positivity persists in a sufficiently small neighborhood.
\end{proof}

For two factors, the determinant of the Sylvester-type differential matrix is nonzero precisely when $\Res(P_1,P_2)\ne0$.  Thus root collisions and coefficient-boundary points are the natural singular loci of the factorization geometry.

\section{Exact low-degree geometry}

Throughout this section let
\[
\Delta_N
=
\left\{
p\in\mathbb R_{\ge0}^{N+1}:
\sum_{n=0}^N p_n=1
\right\}
\]
denote the probability simplex, equipped with the normalized Lebesgue
measure \(\lambda_N\) (equivalently, the Dirichlet\((1,\ldots,1)\)
probability measure).

If \(R_N\subset\Delta_N\) denotes the set of positively reducible
probability polynomials, we write
\[
\mu_N=\lambda_N(R_N)
\]
for its volume fraction.

\subsection{Quadratics}
We start with the elementary  case of degree 2 polynomials.
Write
\begin{equation}
    P(z)=c+bz+az^2,
    \qquad a,b,c\ge0,
    \qquad a+b+c=1.
\end{equation}
A nontrivial positive factorization is necessarily a product of two Bernoulli factors.

\begin{proposition}[Quadratic region]
\label{prop:quadratic}
The polynomial $P$ is positively factorable if and only if
\begin{equation}
    b^2\ge4ac.
    \label{eq:quadratic-disc}
\end{equation}
Equivalently, in coordinates $(a,b)$ with $c=1-a-b$,
\begin{equation}
    2(\sqrt a-a)\le b\le1-a.
    \label{eq:quadratic-region}
\end{equation}
Under the uniform measure on $\Delta_2$, the factorable probability is
\begin{equation}
    \mu_2=\frac13.
\end{equation}
\end{proposition}

\begin{proof}
The first condition is real-rootedness, hence Theorem~\ref{thm:real-rooted}.  Alternatively, write
\begin{equation}
  P(z)=(1-p_1+p_1z)(1-p_2+p_2z).
\end{equation}
Then $a=p_1p_2$ and $b=p_1+p_2-2p_1p_2$.  For fixed $a$, put $p_2=a/p_1$ and minimize over $a\le p_1\le1$.  The minimum occurs at $p_1=\sqrt a$ and equals $2(\sqrt a-a)$; the maximum is $1-a$.  The coordinate area of the factorable region is
\begin{equation}
    \int_0^1\bigl[(1-a)-2(\sqrt a-a)\bigr]da=\frac16,
\end{equation}
while the simplex area is $1/2$.
\end{proof}

\subsection{Cubics}
\begin{proposition}[Cubic criterion and volume]
\label{prop:cubic}
Let
\begin{equation}
  P(z)=p_0+p_1z+p_2z^2+p_3z^3,
  \qquad p_i>0.
\end{equation}
Then $P$ is positively factorable if and only if
\begin{equation}
  p_1p_2\ge p_0p_3.
  \label{eq:cubic-criterion}
\end{equation}
Consequently, under the uniform measure on $\Delta_3$,
\begin{equation}
  \mu_3=\frac12.
\end{equation}
\end{proposition}

\begin{proof}
If $P=(a+bz)(u+vz+wz^2)$ with nonnegative coefficients, then
\begin{align}
  p_0&=au,&p_1&=av+bu,&p_2&=aw+bv,&p_3&=bw,
\end{align}
and
\begin{equation}
  p_1p_2-p_0p_3=v(a^2w+abv+b^2u)\ge0.
\end{equation}
Conversely, let $g(s)=P(-s)$.  We have $g(0)>0$ and
\begin{equation}
  g(p_2/p_3)=p_0-\frac{p_1p_2}{p_3}\le0.
\end{equation}
Hence $P$ has a root $-s$ with $0<s\le p_2/p_3$.  Division gives
\begin{equation}
  P(z)=(z+s)\left[p_3z^2+(p_2-sp_3)z+\frac{p_0}{s}\right],
\end{equation}
and the quotient coefficients are nonnegative.  After normalization this is a positive factorization.  Finally, the coordinate permutation
\begin{equation}
  (p_0,p_1,p_2,p_3)\mapsto(p_1,p_0,p_3,p_2)
\end{equation}
interchanges the two strict inequalities in \eqref{eq:cubic-criterion}; the equality surface has measure zero.
\end{proof}

The criterion is the coefficient form of the cubic Routh--Hurwitz boundary and agrees with Briggs's root-plane classification \cite{Briggs1985}.

\subsection{Hurwitz-stable quartics}
\label{subsec:quartic-hurwitz}

The Routh--Hurwitz criterion gives an exact lower bound on the
factorable volume in degree four.  Let
\[
\mathcal H_4
=
\left\{
p\in\Delta_4^\circ:
\chi_p \text{ is Hurwitz stable}
\right\},
\]
where Hurwitz stability means that every zero lies in the open left
half-plane, and define
\[
\nu_4=\lambda_4(\mathcal H_4).
\]

\begin{proposition}[Volume of the Hurwitz-stable quartic region]
\label{prop:quartic-hurwitz-volume}
Under the uniform probability measure on \(\Delta_4\),
\[
\nu_4=\frac16.
\]
Every Hurwitz-stable quartic probability polynomial is positively
reducible.  Consequently,
\[
\mu_4\ge \nu_4=\frac16.
\]
\end{proposition}

\begin{proof}
Write
\[
P(z)=p_0+p_1z+p_2z^2+p_3z^3+p_4z^4,
\qquad
p_i>0.
\]
The quartic Routh--Hurwitz criterion states that \(P\) is Hurwitz
stable if and only if
\begin{align}
p_3p_2&>p_4p_1,
\label{eq:quartic-rh-first}\\
p_3p_2p_1&>p_4p_1^2+p_3^2p_0.
\label{eq:quartic-rh-second}
\end{align}
For positive coefficients, the second inequality implies the first,
since division by \(p_1>0\) gives
\[
p_3p_2
>
p_4p_1+\frac{p_3^2p_0}{p_1}
>
p_4p_1.
\]
Thus Hurwitz stability is equivalent, up to the measure-zero boundary,
to \eqref{eq:quartic-rh-second}.

By a standard result (see e.g. \cite{KotzDirichlet2000}),  uniformly distributed points of \(\Delta_4\) can be represented as
\[
p_i=\frac{X_i}{X_0+\cdots+X_4},
\]
where \(X_0,\ldots,X_4\) are independent positive random variables with
density
\[
f(x)=e^{-x},
\qquad x\ge0.
\]
Since the Routh--Hurwitz inequalities are homogeneous in
the coefficients, normalization cancels, and therefore
\[
\nu_4
=
\Pr\!\left(
X_3X_2X_1>X_4X_1^2+X_3^2X_0
\right).
\]
Conditioning on \(X_0,X_1,X_3,X_4\), this event is equivalent to
\[
X_2>
\frac{X_4X_1}{X_3}
+
\frac{X_3X_0}{X_1}.
\]
Integrating over \(X_2\) we have,
\[
\Pr\!\left(
P\text{ is Hurwitz stable}
\,\middle|\,
X_0,X_1,X_3,X_4
\right)
=
\exp\!\left(
-\frac{X_4X_1}{X_3}
-\frac{X_3X_0}{X_1}
\right).
\]
Averaging first over \(X_0\) and \(X_4\) yields
\begin{align}
\nu_4
&=
\mathbb E_{X_1,X_3}
\left[
\frac{1}{1+X_3/X_1}
\frac{1}{1+X_1/X_3}
\right]
\nonumber\\
&=
\mathbb E_{X_1,X_3}
\left[
\frac{X_1X_3}{(X_1+X_3)^2}
\right].
\end{align}
For two independent unit-rate exponential variables,
\[
U=\frac{X_1}{X_1+X_3}
\]
is uniformly distributed on \([0,1]\).  Hence
\[
\nu_4
=
\mathbb E[U(1-U)]
=
\int_0^1u(1-u)\,du
=
\frac16.
\]

It remains to relate Hurwitz stability to positive reducibility.
A real Hurwitz-stable polynomial factors over \(\mathbb R\) into
linear factors associated with negative real roots and quadratic
factors associated with nonreal conjugate pairs in the left
half-plane.  Such factors have the forms
\[
z+a,
\qquad a>0,
\]
and
\[
z^2-2\operatorname{Re}(r)z+|r|^2,
\qquad \operatorname{Re}(r)<0,
\]
respectively, and therefore have strictly positive coefficients.
A quartic thus admits a positive factorization into factors of degree
one or two.  After normalization at \(z=1\), this is a positive
factorization into probability polynomials.  Consequently
\(\mathcal H_4\subseteq\mathcal R_4\), and
\[
\mu_4\ge\nu_4=\frac16.
\]
\end{proof}

\begin{remark}
The Monte Carlo estimate
\[
\mu_4\simeq 0.43
\]
is substantially larger than \(1/6\).  Thus the Hurwitz-stable sector
provides a rigorous positive-volume subset of the factorable quartics,
but accounts for only part of the full factorable region.
\end{remark}

\section{Equality cases, direct sum decompositions and sparse examples}

This section illustrates two complementary aspects of the theory.
The first is an arithmetic family of probability polynomials admitting
canonical positive factorizations arising from mixed-radix
representations of integers.
The second is the observation that these factorizations are naturally
distinguished as equality cases of the entropy inequality of
Theorem~3.2, rather than as entropy-maximizing factorizations.
Indeed, Proposition~3.5 shows that maximal positive atomizations need
not have equal entropy, so equality in Theorem~3.2 should not be
confused with optimality in the definition of
\(S_{\rm fac}\).

Let

\[
C_N(z)=1+z+\cdots+z^N.
\]

\begin{proposition}
Suppose
\[
N+1=ab,
\qquad
a,b>1.
\]
Then
\[
C_{ab-1}(z)
=
(1+z+\cdots+z^{a-1})
(1+z^a+z^{2a}+\cdots+z^{(b-1)a}).
\]
After normalization, this factorization is an equality case of
Theorem~3.2.
\end{proposition}

\begin{proof}
Every integer
\[
0\le n<ab
\]
has a unique Euclidean decomposition
\[
n=r+as,
\qquad
0\le r<a,
\qquad
0\le s<b.
\]
Consequently every monomial appears exactly once in the product.

After normalization, let
\(A\)
be uniformly distributed on
\(\{0,\ldots,a-1\}\),
and
\(B\)
uniformly distributed on
\(\{0,\ldots,b-1\}\).
The factorization corresponds precisely to the decomposition

\[
X=A+aB.
\]

Since the map

\[
(A,B)\mapsto A+aB
\]

is injective,

\[
H(A+a B)=H(A,B)=H(A)+H(B),
\]

so the factorization realizes equality in
Theorem~3.2.
\end{proof}

The preceding proposition should be interpreted as exhibiting a
particularly simple equality case of the entropy inequality.
It does \emph{not} imply that this factorization maximizes
\(S_{\rm fac}\).
Indeed, Proposition~3.5 shows that distinct maximal atomizations of the
same polynomial may carry different entropies.
Thus equality in Theorem~3.2 is a property of a particular
factorization rather than of the probability polynomial itself.


We remark that the polynomials ${C_N(z)}$ are special and  illustrate the distinction between density
in the degree parameter and volume inside the probability simplex.
Whenever
\(N+1\)
is composite, the polynomial 
\[
\frac{C_N(z)}{N+1}
\]
admits the above positive factorization.
Since the integers with
\(N+1\)
composite have density one,
\[
\frac{
\#\{1\le N\le M:\;N+1\text{ composite}\}
}{M}
\longrightarrow
1,
\]
the ${C_N(z)}$ polynomials are factorable for almost every degree.

This arithmetic statement should not be confused with the geometric question of the simplex volume
$\mu_N$, studied in the next section.
The polynomial
\(C_N/(N+1)\)
is merely the barycenter of the simplex, whereas a typical Dirichlet
sample fluctuates on the scale
\(N^{-1}\)
in each coordinate.
Accordingly, density-one factorability of the barycenter does not imply
that a typical probability polynomial is positively reducible.

Finally, the factorization above is closely related to factorizations of
finite cyclic groups and mixed-radix tilings of intervals, connecting
the present theory with the arithmetic literature on cyclotomic
factorizations and tilings
\cite{deBruijn1953,Steinberger2012}.

\section{Random probability polynomials}
We now investigate the volume fractions
\(\mu_N\)
and
\(\nu_N\)
numerically under the measure
\(\lambda_N\).
The exact values have been already given $\mu_2=1/3$ and $\mu_3=1/2$.

\subsection{Endpoint constraints}

For $p\sim\mathrm{Dirichlet}(1,\ldots,1)$, each coordinate has the beta marginal
\begin{equation}
  p_0\sim\mathrm{Beta}(1,N),
  \qquad
  \Pr(p_0>t)=(1-t)^N.
  \label{eq:beta-tail}
\end{equation}
This follows directly by slicing the simplex at fixed $p_0=t$: the remaining simplex has volume proportional to $(1-t)^{N-1}$; equivalently it is the standard one-coordinate marginal of a Dirichlet law \cite{KotzDirichlet2000}.  In particular,
\[
\Pr(Np_0>x)
=
\left(1-\frac{x}{N}\right)^N
\longrightarrow
e^{-x},
\qquad x\ge0.
\]

If $P=\prod_{j=1}^mP_j$ and $q_{j,0}=P_j(0)$, then
\begin{equation}
  p_0=\prod_{j=1}^mq_{j,0},
  \qquad
  -\log p_0=\sum_{j=1}^m-\log q_{j,0}.
  \label{eq:endpoint-product}
\end{equation}
Thus a factorization of a typical simplex point must satisfy
\begin{equation}
  \sum_{j=1}^m-\log q_{j,0}=\log N+O_{\mathbb P}(1),
  \label{eq:endpoint-scale}
\end{equation}
with $O(1)$ fluctuations. This is a necessary scale constraint, not a criterion for factorability.  A bounded number of factors forces at least one $q_{j,0}$ to be polynomially small in $N$; a logarithmic number of order-one factors is compatible with the typical scale; and a linear number of balanced Bernoulli-like factors generally gives an exponentially small endpoint coefficient, which lies in a lower simplex tail.  The identical statement holds for the leading coefficient.

\subsection{Monte Carlo estimates}

Uniform simplex samples were generated by normalizing independent unit-rate exponential random variables.  For each sample, the roots were grouped into real roots and conjugate pairs, all conjugation-closed subsets were enumerated, and both a candidate factor and its complement were tested for coefficient nonnegativity.  Table~\ref{tab:mc} and Fig.~\ref{fig:mc} report the resulting estimates.  The exact values in degrees two and three provide internal checks.

\begin{table}[t]
\centering
\small
\begin{tabular}{r r c c}
\toprule
$N$ & samples & $\widehat\mu_N$ & 95\% interval\\
\midrule
2&30000&0.3301&[0.3248,0.3354]\\
3&30000&0.4974&[0.4917,0.5031]\\
4&15000&0.4302&[0.4223,0.4381]\\
5&15000&0.4226&[0.4147,0.4305]\\
6&12000&0.3759&[0.3673,0.3846]\\
7&12000&0.3288&[0.3203,0.3372]\\
8&10000&0.2855&[0.2766,0.2944]\\
9&10000&0.2528&[0.2443,0.2613]\\
10&8000&0.2241&[0.2150,0.2333]\\
11&6000&0.1993&[0.1892,0.2094]\\
12&5000&0.1820&[0.1713,0.1927]\\
13&4000&0.1600&[0.1486,0.1714]\\
14&3000&0.1530&[0.1401,0.1659]\\
\bottomrule
\end{tabular}
\caption{Estimated factorable volume under the uniform measure on $\Delta_N$.  Intervals are normal-approximation binomial intervals.}
\label{tab:mc}
\end{table}

\begin{figure}[t]
  \centering
  \includegraphics[width=.58\textwidth]{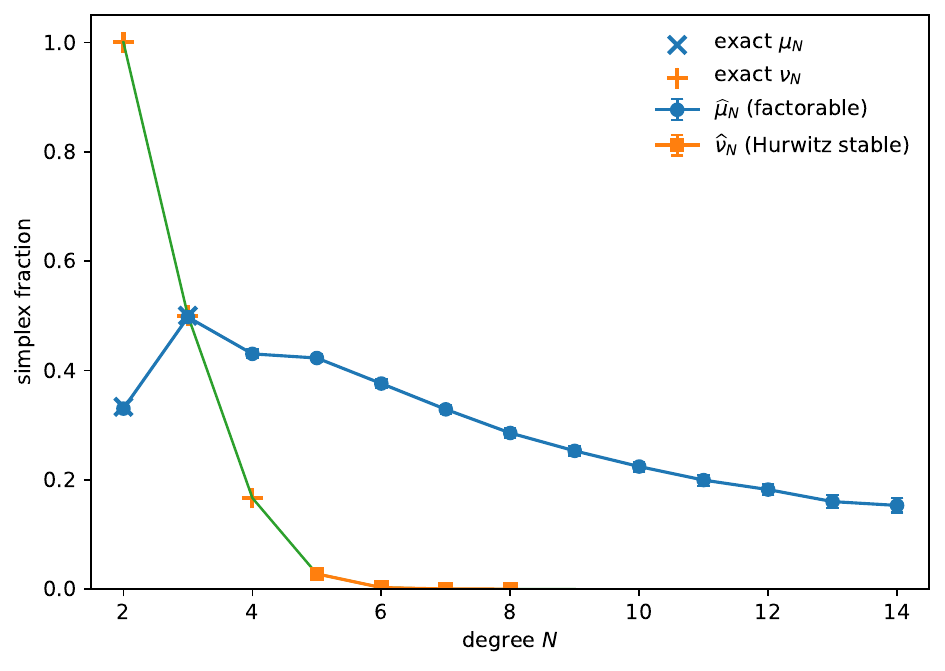}
  \caption{Monte Carlo estimates of $\mu_N$ and $\nu_N$ the fraction of factorizable probability polynomials and the subset of Hurwitz polynomials on the probability simplex.  The cross markers show the exact values. For $N=4$ the exact result $\nu_4=1/6$ is described in proposition~\ref{prop:quartic-hurwitz-volume}}
  \label{fig:mc}
\end{figure}

The data suggest decay after a low-degree maximum, but they do not distinguish a power law from slower alternatives.  The ensemble is a noncentered positive-coefficient random-polynomial ensemble, so work on real zeros of noncentered random polynomials is relevant background, although positive reducibility is a stronger event than the existence of negative real roots \cite{DoNguyenVu2018,Do2019}.

\subsection{Moment-constrained ensembles}
\label{subsec:constrained-ensembles}

The uniform measure on the full simplex treats all probability vectors
equally, independently of the location and width of the corresponding
counting distribution.  Positive factorizations, however, need not be
distributed uniformly with respect to these macroscopic observables.
It is therefore natural to compare the full-simplex ensemble with
ensembles conditioned on the first two moments.

For fixed mean \(m\) and variance \(v\), define the moment-constrained
polytope
\[
\Delta_N(m,v)
=
\left\{
p\in\Delta_N:
\sum_{n=0}^{N} n p_n=m,\qquad
\sum_{n=0}^{N}(n-m)^2p_n=v
\right\}.
\]
Equivalently, once \(m\) is fixed, the constraints are affine:
\[
\sum_{n=0}^{N}p_n=1,\qquad
\sum_{n=0}^{N}np_n=m,\qquad
\sum_{n=0}^{N}n^2p_n=m^2+v.
\]
Whenever this polytope has nonempty relative interior, let
\(\lambda_N^{m,v}\) denote its normalized Hausdorff measure and define
\[
\mu_N(m,v)
=
\lambda_N^{m,v}
\left\{
p:\chi_p\ \text{is positively reducible}
\right\}.
\]
We similarly define
\[
\nu_N(m,v)
=
\lambda_N^{m,v}
\left\{
p:\chi_p\ \text{is Hurwitz stable}
\right\}.
\]
The latter quantity probes the portion of the constrained ensemble
belonging to the classical Hurwitz sector.  For \(N\ge3\),
Hurwitz stability implies positive reducibility, and hence
\[
\nu_N(m,v)\le \mu_N(m,v).
\]

We concentrate on distributions centered in the support,
\[
m=\frac{N}{2},
\]
and compare two width scalings:
\[
v=1
\qquad\text{and}\qquad
v=\frac{N}{2}.
\]
The first ensemble remains narrowly concentrated as the support grows,
whereas the second has a width of order \(\sqrt{N}\), as occurs for
sums of a number of weakly fluctuating contributions proportional to
\(N\).

\subsubsection{Monte Carlo on affine slices}

Uniform samples from
\(\Delta_N(N/2,v)\)
were generated by a hit-and-run Markov chain in the affine subspace
defined by the three moment constraints.  Starting from an interior
point, a random direction was drawn in the null space of the constraint
matrix and the next point was sampled uniformly from the maximal line
segment that remained inside the simplex.  After burn-in, samples were
retained at fixed thinning intervals.

Positive reducibility was tested by computing the root multiset,
grouping nonreal roots into conjugate pairs, and enumerating
conjugation-invariant root subsets as in
remark~\ref{rem:algorithms}.
Hurwitz stability was tested directly from the location of the roots.
The results are shown in Fig.~\ref{fig:constrained-ensembles}.

\begin{figure}[t]
    \centering
    \includegraphics[width=0.78\linewidth]{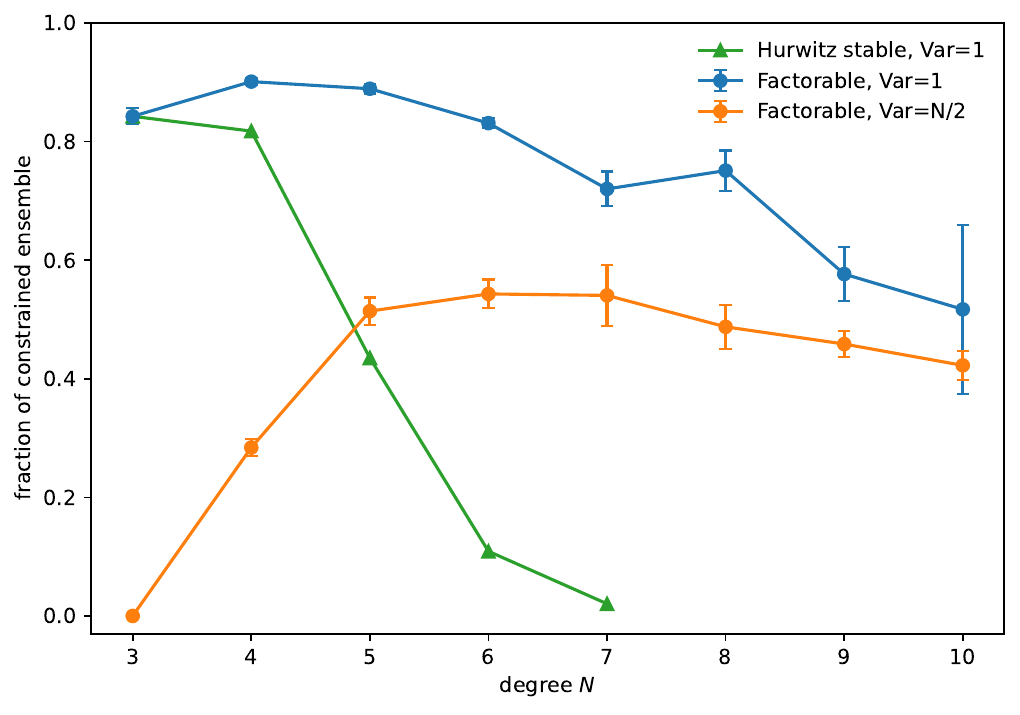}
    \caption{
    Estimated fractions of positively reducible polynomials in the
    centered moment-constrained ensembles
    \(\operatorname{Var}(n)=1\) and
    \(\operatorname{Var}(n)=N/2\).
    The third curve shows the Hurwitz-stable fraction in the
    \(\operatorname{Var}(n)=1\) ensemble. Error bars reflect the variation among independent Markov chains.
    Hurwitz probabilities for
    \(\operatorname{Var}(n)=N/2\) are not shown because no Hurwitz-stable
    samples were observed in the present runs.
    }
    \label{fig:constrained-ensembles}
\end{figure}

The constrained ensembles show substantially greater positive
reducibility than the unconditioned uniform-simplex ensemble over much
of the range studied.  For the narrow ensemble \(v=1\), the estimated
factorable fraction is close to \(0.9\) at degrees \(N=4,5\), and
remains appreciable through degree \(12\).  For
\(v=N/2\), positive reducibility is approximately one half for
\(5\le N\le8\), followed by a slower decrease at larger degrees.
Thus concentration near the center of the support strongly favors
positive factorability, although the numerical data do not suggest
that the factorable probability remains close to one asymptotically.

The Hurwitz-stable part behaves quite differently.  In the narrow
ensemble \(v=1\), its estimated probability decreases from approximately
\(0.85\) at \(N=3\) to approximately \(0.02\) at \(N=7\), with no
Hurwitz-stable samples observed for \(N\ge8\) in the present runs.
For the broader \(v=N/2\) ensemble, no Hurwitz-stable samples were
observed over the sampled range.  This contrast indicates that the
enhanced positive factorability of centered distributions is not
explained solely by the classical Hurwitz sector: at moderate and large
degree, most observed positive factorizations arise from root groupings
that extend beyond Hurwitz-stable polynomials.

These calculations should be regarded as exploratory.  Hit-and-run
samples are correlated, so the binomial error bars shown in
Fig.~\ref{fig:constrained-ensembles} do not include uncertainty due to
finite mixing time.  A higher-precision study should use several
independent chains, estimate integrated autocorrelation times, and test
the stability of the factorability classification under increased
numerical precision.  Nevertheless, the present data support the
qualitative conclusion that factorability depends strongly on the
macroscopic shape of the probability law and motivate the study of
\(\mu_N(m,v)\) as a function of the imposed moments.

\paragraph{Monte Carlo methodology}
For the unconstrained ensemble, probability vectors were sampled
uniformly from the simplex by normalizing independent unit-rate
exponential random variables.  For the moment-constrained ensembles,
uniform samples from the affine polytope
\(
\Delta_N(m,v)
\)
were generated using a hit-and-run Markov chain.  Starting from an
interior feasible point, each step selects a random direction in the
null space of the constraints and moves to a point chosen
uniformly along the maximal line segment contained in the constrained
polytope.  
fixed thinning intervals.

Positive reducibility was tested by computing the roots of the
probability-generating polynomial, grouping nonreal roots into
conjugate pairs, enumerating all conjugation-invariant root subsets,
and checking whether the corresponding factors could be normalized to
probability polynomials with nonnegative coefficients.  Hurwitz
stability was tested independently by verifying that all roots lie in
the open left half-plane.

To assess the reliability of the constrained-ensemble calculations, we
performed several independent hit-and-run chains initialized from
widely separated points.  The reported uncertainties are based
on the variation between independent chain averages rather than on
naive binomial counting statistics.  In addition, integrated
autocorrelation times of the factorability indicator were monitored to
estimate effective sample sizes, and the numerical classification was
checked under several coefficient-positivity tolerances to identify
samples lying close to the factorization boundary.  Our exact results for
low degrees were used as internal benchmarks for the
factorability algorithm.

Although the numerical study is exploratory, all qualitative
conclusions reported were found to be robust under changes of
burn-in length, thinning interval, and positivity tolerance.

\begin{remark}[Algorithms and complexity]
The search for factorization of a fixed polynomial proceeds with brute force as follows: compute the root multiset, form possible factors by going through the partitioning of the set of roots into sets and checking if this results in a product of positive integer polynomials. \label{rem:algorithms}
\end{remark}

If $r$ is the number of real roots and $c$ the number of nonreal conjugate pairs, the naive subset enumeration has size $2^{r+c}$.  Repeated roots, exact arithmetic, and coefficient certification require care.  For floating-point data, interval arithmetic or a backward-error certificate should accompany a claimed boundary factorization.

\section{Open problems}

The present results suggest several concrete directions.
\begin{enumerate}[leftmargin=2em]
\item Determine the asymptotic behavior of $\mu_N$ under the uniform simplex measure and under other Dirichlet laws.
\item Extend the theory to analytic probability-generating functions of infinite-support distributions.
\item Develop the multivariate theory for vector-valued counts, where factor supports define a hypergraph of latent interactions.
\item Determine whether $C_{p-1}(z)$ is a positive atom for every prime $p$ over real nonnegative coefficients; the corresponding integer-coefficient question is linked to cyclotomic and tiling theory.

\item Determine the asymptotic behavior of the conditioned volumes
\[
\mu_N\!\left(\frac N2,v_N\right)
\]
for fixed \(v_N\), for \(v_N\asymp N\), and for other natural
moment scalings.

\end{enumerate}

\section{Conclusion}

The positive factorization poset provides a new combinatorial object naturally associated with a probability-generating function. The entropy studied here is only one invariant of this object; many others remain to be explored. Factorization entropy is a monotone functional on this poset, maximized on positive atomizations and equal to the counting entropy precisely when the latent addition map is injective.  The root structure controls the real-rooted and Hurwitz sectors, additive supports provide coefficient-free obstructions, and multiplication maps provide a local geometric description of stable factorizations.  The exact low-degree volumes and numerical simplex study indicate that positive factorability has a nontrivial geometry already in modest degree.  These results place aggregate-count inference, binding-polynomial factorization, and related counting problems within a common algebraic-probabilistic framework.

{\it Use of AI-Assisted Tools}.
Large language models (ChatGPT) assisted with text editing and code development. All mathematical derivations, scientific claims, and code were independently verified by the author.

\bibliographystyle{elsarticle-num}
\bibliography{positive_factorizations}

\end{document}